\documentclass[12pt]{article}%
\usepackage{amsfonts}
\usepackage{amsmath}
\usepackage{amssymb}
\usepackage{graphicx}%
\providecommand{\U}[1]{\protect\rule{.1in}{.1in}}
\newtheorem{theorem}{Theorem}

\newtheorem{condition}[theorem]{Condition}

\newtheorem{corollary}[theorem]{Corollary}

\newtheorem{definition}[theorem]{Definition}
\newtheorem{example}[theorem]{Example}

\newtheorem{lemma}[theorem]{Lemma}
\newtheorem{notation}[theorem]{Notation}

\newtheorem{proposition}[theorem]{Proposition}
\newtheorem{remark}[theorem]{Remark}

\newenvironment{proof}[1][Proof]{\textbf{#1.} }{\ \rule{0.5em}{0.5em}}

\begin{document}

\title{Money as Minimal Complexity\thanks{In honor of Lloyd Shapley. }}
\author{Pradeep Dubey\thanks{Stony Brook Center for Game Theory, Dept. of Economics;
and Cowles Foundation for Research in Economics, Yale University}, Siddhartha
Sahi\thanks{Department of Mathematics, Rutgers University, New Brunswick, New
Jersey}, and Martin Shubik\thanks{ Cowles Foundation for Research in
Economics, Yale University; and Santa Fe Institute, New Mexico.}}
\date{December 16, 2015}
\maketitle

\section*{Abstract}

We consider mechanisms that provide traders the \emph{opportunity} to exchange
commodity $i$ for commodity $j$, for certain ordered pairs $ij.$ Given any
connected graph $G$ of opportunities, we show that there is a unique mechanism
$M_{G}$ that satisfies some natural conditions of \textquotedblleft
fairness\textquotedblright\ and \textquotedblleft
convenience\textquotedblright. Let $\mathfrak{M}(m)$ denote the class of
mechanisms $M_{G}$ obtained by varying $G$ on the commodity set $\left\{
1,\ldots,m\right\}  $. We define the complexity of a mechanism $M$ in
$\mathfrak{M(m)}$ to be a certain pair of integers $\tau(M),\pi(M)$ which
represent the time\ required to exchange $i$ for $j$ and the
information\ needed\ to determine the exchange ratio (each in the worst case
scenario, across all $i\neq j$). This induces a quasiorder $\preceq$ on
$\mathfrak{M}(m)$ by the rule%
\[
M\preceq M^{\prime}\text{ if }\tau(M)\leq\tau(M^{\prime})\text{ \emph{and}
}\pi(M)\leq\pi(M^{\prime}).
\]

We show that, for $m>3$, there are precisely three $\preceq$-minimal
mechanisms $M_{G}$ in $\mathfrak{M}(m)$, where $G$ corresponds to the star,
cycle and complete graphs. The star mechanism has a distinguished commodity --
the money -- that serves as the sole medium of exchange and mediates trade
between decentralized markets for the other commodities.

Our main result is that, for \emph{any} weights $\lambda,\mu>0,$ the star
mechanism is the \emph{unique} minimizer of $\lambda\tau(M)+\mu\pi
(M\mathcal{)}$ on $\mathfrak{M}(m)$ for large enough $m$.

\textbf{JEL Classification}: C70, C72, C79, D44, D63, D82.

\textbf{Keywords: }exchange mechanism, minimal complexity, money.

\section{Introduction}

The need for money in an exchange mechanism has been the topic of much
discussion, and it would be impossible to summarize that literature here. We
give some references that are indicative, but by no means exhaustive. (For a
detailed survey, see \cite{Shubik:1993} and \cite{starr:2012}.)

Several search-theoretic models, involving random bilateral meetings between
long-lived agents, have been developed following Jevons \cite{Jevons:1875}
(see, \textit{e.g.}, \cite{Bannerjee-Maskin(1996)}, \cite{Iwai:1996},
\cite{Jones:1976}, \cite{Kiyotaki-Wright: 1989}, \cite{Kiyotaki-Wright: 1993},
\cite{Li-Wright:1998}, \cite{Ostroy 1973}, \cite{Trejos-Wright: 1995} and the
references therein). These models turn on utility-maximizing behavior and
beliefs of the agents in Nash equilibrium, and shed light on which commodities
are likely to get adopted as money. A parallel, equally distinctive, strand of
literature builds on partial or general equilibrium models with other kinds of
frictions in trade, such as limited trading opportunities in each period, or
transaction costs (see, \textit{e.g.}, \cite{Foley:1970}, \cite{Hahn:1971},
\cite{Heller: 1974}, \cite{heller-Starr: 1976}, \cite{Howitt-Clower: 2000},
\cite{Ostroy-Starr: 1974}, \cite{Ostroy-Starr: 1990}, \cite{starr:2012},
\cite{Starret:1973}, \cite{Wallace: 1980}). In many of these models, a
specific trading mechanism is fixed exogenously, and the focus is on activity
within the mechanism that is induced by equilibrium, based again on the
optimal behavior of utilitarian individuals.

Our approach complements this literature in two salient ways, and brings to
light a new rationale for money that is different from those proposed earlier,
but not at odds with them, in that the door is left open to incorporate their
concerns within our framework. First and foremost, our focus is purely on
mechanisms of trade with no regard to the characteristics of the individuals
such as their endowments, production technologies, preferences or beliefs.
Second, no specific mechanism is\emph{ }specified \textit{ex ante} by us. We
start with a welter of mechanisms and cut them down by four natural conditions
and certain complexity criteria, ultimately ending up with the
\textquotedblleft star\textquotedblright\ mechanism in which money plays the
central role.

These mechanisms are \emph{Cournotian }in spirit\footnote{It is our purpose to
see how far matters may develop in an elementary Cournot framework. In
particular, note that \textit{ex ante }there are no \textquotedblleft
prices\textquotedblright\ to refer to, upon which a trader may condition his
offers. We do show that prices can be \textquotedblleft
admitted\textquotedblright, i.e., \emph{defined}, but this happens \textit{ex
post} once unconditional offers for trade have come into the mechanism. Our
mechanisms are thus a far cry from the more complex Bertrand mechanisms, in
which traders use prices alongside quantities in order to make contingent
statements to protect themselves against vagaries of the market (see,e.g.,
\cite{Dubey:1982}, \cite{Mertens:2003}). An analysis analogous to ours might
well be possible in the Bertrand setting, but that is a topic for future
exploration.}\emph{, }and the setting for them is simple, in keeping with our
aim of showing that the need for money can arise at a very rudimentary level.
A mechanism $M$ on commodity set $\left\{  1,\ldots,m\right\}  $ operates as
follows. For certain ordered pairs $ij,$ pre-specified by $M$, each trader may
offer any quantity of commodity $i$ in order to obtain commodity $j$. Once all
offers are in, the mechanism $M$ redistributes to the traders the commodities
it has received, holding back nothing. The returns to the traders are
calculated by an algorithm\footnote{There is no\textit{ }presumption that the
algorithm be \textquotedblleft informationally decentralized\textquotedblright%
. Indeed even the return to a simple offer of $i$, made only via the pair
$ij$, may\emph{ }well\emph{ }depend on all the offers at \emph{every} $kl\in
G;$ and may thus require a lot of information for its computation.} that is
common knowledge. Thus a mechanism $M$ is characterized by a collection of
exchange \emph{opportunities }$ij$, which form the edges of a directed graph
$G$ on nodes $\left\{  1,\ldots,m\right\}  $, and the algorithm. We assume
throughout that $G$ is \emph{connected, }i.e., $M$ permits \emph{iterative}%
\textit{ }exchange of any $i$ for any $j.$

At this level of generality, there are infinitely many mechanisms (algorithms)
for any given graph $G$. However, we shall show that only one of them
satisfies some natural conditions of \textquotedblleft
fairness\textquotedblright\ and \textquotedblleft
convenience\textquotedblright\ (see Section \ref{characterization}). This
special mechanism is denoted $M_{G}$ and is described precisely in Section
\ref{Formal model}. It is a striking property of $M_{G}$ that it admits unique
prices\footnote{Prices are to be thought of as consistent exchange rates
between commodities, i.e. the ratios $p_{i}/p_{j}$. Thus they correspond to
rays in $\mathbb{R}_{++}^{m}$, each of which is represented by a vector $p$ in
$\mathbb{R}_{++}^{m}$ (identified with all its scalar multiples $\lambda p$
for $\lambda>0$).
\par
{}}, which depend only on the aggregate offers by the traders on the various
edges of $G$, and which mediate trade in the following strong sense: first,
the return to any trader depends only on his own offers and the prices;
second, the total value --- under the prevailing prices --- of every trader's
offers is equal to that of his returns. The immediate upshot of price
mediation is that the returns to any trader can be calculated in a transparent
manner from the prices and his own offers.

Thus we are led to consider the class $\mathfrak{M}(m)$ of mechanisms $M_{G}$,
where $G$ ranges over all directed, connected graphs on the vertex set
$\left\{  1,\ldots,m\right\}  $. The cardinality of $\mathfrak{M}(m),$ though
finite, grows super-exponentially in $m.$ However we shall show in Section
\ref{Formal model} that if one invokes natural complexity considerations,
based on the time needed to exchange any commodity $i$ for $j$ and the
information needed to determine the exchange ratio $p_{i}/p_{j}$, then the
welter of mechanisms in $\mathfrak{M}(m)$ is eliminated and we are left with
only three mechanisms of minimal complexity, namely those that arise from the
star, cycle and complete graphs (Theorem \ref{Emergence of money}). Indeed,
provided $m$ is large enough, just the star\ mechanism remains (Theorem
\ref{Uniqueness of money}) in which one commodity emerges endogenously as
money and mediates trade across decentralized markets for the other
commodities\footnote{To be precise: the price of any commodity $1\leq i\leq
m-1$ , in terms of money $m,$ depends only on the aggregate offers on edges
$im$ and $mi$; and thus this pair of edges may be viewed as a decentralized
market for $i$ and $m,$ with $m$ mediating between the various markets.}.

Our analysis is carried out in the oligopolistic setting of finitely many
traders. However, in Section \ref{continuum} we show that it readily extends
to the case of \textquotedblleft perfect competition\textquotedblright, where
there is a continuum of traders and $G$-mechanisms induce \textquotedblleft
price-taking\textquotedblright\ behavior as in the Walrasian\ model .

It is worth emphasizing that ours is a purely \textquotedblleft
mechanistic\textquotedblright, as opposed to a \textquotedblleft
utilitarian\textquotedblright\ or \textquotedblleft
behavioral\textquotedblright, approach to the emergence of
money\footnote{There \emph{is} a faint touch of rationality that we assume
regarding the traders, but it an order-of-magnitude milder than utilitarian
(or other behavioral) considerations. See Remark \ref{touch of rationality}.}.
In the parlance of game theory, we are concerned with the \textquotedblleft
game form\textquotedblright\ behind the game or --- to be more precise ---
with the mechanism that underlies the game form itself. Indeed, with the same
mechanism as the foundation, several different game forms can be constructed
by introducing other considerations, such as whether netting\footnote{Netting
means that if an individual \emph{ex ante} offers $x$ units of commodity $i$
to the mechanism, and is \emph{ex post} entitled to receive $y$ units of $i$
from it, then he is deemed to owe $\max\left\{  0,x-y\right\}  $ or else to
receive $\max\left\{  0,y-x\right\}  .$ In this scenario, one may think that
\textquotedblleft offers\textquotedblright\ consist of \emph{ promises} to
deliver commodities, rather than commodities themselves; and that the
mechanism calls upon traders to make (take) net deliveries (receipts) of
actual commodities. But note the \textit{a priori} need for a mechanism with
respect to which netting can be formulated (or, for that matter, borrowing and
default, or any other trade regulation). Also note that these these
regulations do not come without a cost (see Remark \ref{clearinghouse}).} of
commodities is permitted or not, and if so to what extent; or whether certain
commodities can be borrowed prior to trade and on what terms, along with rules
for the settlement of debt in the event of default. These are no doubt
important economic issues, bearing on the \textquotedblleft
liquidity\textquotedblright\ in the system and the efficiency of its
equilibria. They have been discussed at length, often in terms of the star
mechanism which conforms to the well-known Walrasian model once there is
perfect competition\ and \textquotedblleft sufficient
liquidity\textquotedblright\ (see, e.g., \cite{Dubey-Shapley: 1994},
\cite{Dubey-Geanakoplos-Shubik}, \cite{Dubey-Geanakoplos ISLM}). However, to
even raise these issues, we first \emph{need a mechanism in the background}.
It is this background \emph{alone} that forms the domain of our inquiry .

Our analysis builds squarely upon \cite{Dubey-Sahi: 2003}, which provided an
axiomatic characterization of the finite set of "G-mechanisms" (see Section
\ref{Formal model}), bridging the gap between the Shapley-Shubik model of
decentralized \textquotedblleft trading posts\textquotedblright, i.e., the
star mechanism (see \cite{Shapley: 1976}, \cite{Shapley-Shubik: 1977},
\cite{Shubik: 1973}) and the Shapley model of centralized \textquotedblleft
windows\textquotedblright, i.e., the complete mechanism\ (see
\cite{Sahi-Yao:1989}). Various strategic market games, based upon trading
posts, have been analyzed, with commodity or fiat money in
\cite{Dubey-Shubik:1978}, \cite{Peck:1992}, \cite{Peck-Shell:1992},
\cite{Postlewaite:1978}, \cite{Shapley: 1976}, \cite{Shapley-Shubik: 1977},
\cite{Shubik: 1973}, \cite{Shubik-Wilson}; most of these papers also discuss
the convergence of Nash equilibria (NE) to Walras equilibria (WE) under
replication of traders. For a continuum-of-traders version, with details on
explicit properties of the commodity money (its distribution and desirability)
or of fiat money (its availability and the harshness of default penalties),
under which we obtain equivalence (or near-equivalence) of NE and WE, see
\cite{Dubey-Shapley: 1994}, \cite{Dubey-Geanakoplos-Shubik},
\cite{Dubey-Geanakoplos ISLM}; and, for an axiomatic approach to the
equivalence phenomenon, see \cite{Dubey-MasColell-Shubik:1980}.

Strategic market games differ in a fundamental sense from the Walrasian model,
despite the equivalence of NE and WE. In the WE framework, agents always
optimize generating supply and demand, but markets do not clear except at
equilibrium. We are left in the dark as to what happens outside of
equilibrium. In sharp contrast, in the NE\ framework, markets always clear,
producing prices and trades based on agents' strategies; but agents do not
optimize except at equilibrium. The very formulation of a game demands that
the \textquotedblleft game form\textquotedblright, \textit{i.e.}, the map from
strategies to outcomes, be defined prior to the introduction of agents'
preferences on outcomes; thus disentangling the physics of trade from its
psychology\footnote{To put it bluntly, the insistence on a game form pertains
to the following situation in the real world. People exercise choice all the
time through their actions; and the world goes merrily on, by well-defining
the outcome of those actions --- it does not come to a standstill until they
can explain why they have acted as they did!}. Our mechanisms are firmly in
this genre, and indeed form the bases upon which many market games are built.
To be precise: game forms arise from our mechanisms by introducing private
endowments and the rules of trade (including the degree of netting or
borrowing permitted); and strategic market games then arise by further
introducing preferences.

\section{The Emergence of Money\label{Formal model}}

Let $G$ be a directed and connected graph\footnote{In this paper by a graph we
mean a \emph{directed simple} \emph{graph}. Such a graph $G$ consists of a
finite \emph{vertex} set $V_{G}$, togther with an \emph{edge} set
$E_{G}\subseteq V_{G}\times V_{G}$ that does not contain any loops,
\textit{i.e.,} edges of the form $ii$. For simplicity we shall often write
$i\in G$, $ij\in G$ in place of $i\in V_{G}$, $ij\in E_{G}$ but there should
be no confusion. By a \emph{path} $ii_{1}i_{2}\ldots i_{k}j$ from $i$ to $j$
we mean a nonempty sequence of edges in $G$ of the form $ii_{1},i_{1}%
i_{2},\ldots,i_{k-1}i_{k},i_{k}j.$ If $k=0$ then the path consists of the
single edge $ij$, otherwise we insist that the \emph{intermediate }vertices
$i_{1},\ldots,i_{k}$ be distinct from each other and from the endpoints $i,j$.
However we do allow $i=j$, in which case the path is called a \emph{cycle}. We
say that $G$ is \emph{connected} if for any two vertices $i\neq j$ there is a
\emph{path} from $i$ to $j$.} with vertex set $\left\{  1,\ldots,m\right\}  $.
We define a mechanism $M_{G}$ as follows. Each trader can use every
opportunity in $M,$ i.e., place arbitrary weights on the edges $ij$ of $G,$
representing his offer of $i$ for $j.$ Let $b_{ij}$ denote the total weight on
$ij$ (i.e., the aggregate amount of commodity $i$ offered for $j$ by all
traders). We shall specify what happens when $b_{ij}>0$ for every edge $ij$ in
$G$, i.e., when there is sufficient diversity in the population of traders so
that each opportunity is active. Denote $b=(b_{ij})$ and let $\mathbb{R}%
_{++}^{m}/\sim$ be the set of rays in $\mathbb{R}_{++}^{m}$ representing
prices. It is well-known that (with $b_{ij}$ understood to be $0$ if $ij$ is
not an edge in $G$) there is a unique ray $p=p(b)$ in $\mathbb{R}_{++}%
^{m}/\sim$ satisfying%
\begin{equation}
\sum\nolimits_{i}p_{i}b_{ij}=\sum\nolimits_{i}p_{j}b_{ji}\text{ for all }j.
\label{=price}%
\end{equation}
Note that the left side of (\ref{=price}) is the total value of all the
commodities \textquotedblleft chasing\textquotedblright\ $j$, while the right
side is the total value of commodity $j$ on offer; thus (\ref{=price}) is
tantamount to \textquotedblleft\emph{value conservation}\textquotedblright.

It turns out that (\ref{=price}) has an explicit combinatorial solution, which
we now describe. Let $\mathcal{T}_{i}$ be the collection of all
\textquotedblleft spanning\textquotedblright\ trees in $G$ that are rooted at
$i$ (\emph{i.e.} subgraphs of $G$ in which there is a unique directed path to
$i$ from \emph{every} $j\neq i$); and for any subgraph $H$, define
$b_{H}=\prod\nolimits_{ij\in H}b_{ij}$; then we have\footnote{Formula
(\ref{=priceformula}) has a short proof \cite{Sahi:2013} but a long history.
It seems to be orginally due to \cite{Hill:1966} but has been rediscovered
several times (see the discussion in \cite{Anantharam-Tsoucas: 1989}).}
\begin{equation}
p_{i}=\sum\nolimits_{T\in\mathcal{T}_{i}}b_{T}\text{.} \label{=priceformula}%
\end{equation}

The principle of value conservation, which determines prices, also determines
trade. An individual who offers $a_{ij}$ units of $i$ via opportunity $ij$
gets back $r_{j}$ units of $j$, where $p_{i}a_{ij}=p_{j}r_{j.}$ More
generally, if a trader offers $a=(a_{ij})\geq0$ across all edges of $G$, he
gets a return $r(a,b)\in$ $\mathbb{R}_{+}^{m}$ whose components are given by
\begin{equation}
r_{j}(a,b)=\sum\nolimits_{i}(p_{i}/p_{j})a_{ij} \label{MGr}%
\end{equation}
for all $j.$ This completes the definition of the $G$-\emph{mechanism}%
\footnote{$G$-mechanisms may arise naturally in the context of currency
exchange, with edges $ij$ indicating the direct convertibility of currency $i$
to currency $j.$} $M_{G}.$

Note that the return to a trader depends only on his offer $a$ and the
\emph{price ratios }$p_{i}/p_{j},$ which are well-defined functions of $b$
(unlike the price vector $p=(p_{i})$ which is only defined up to a scalar
multiple). It might be instructive to see the formulae for price ratios (and
thereby also for returns, thanks to equation (\ref{MGr})) for specific
mechanisms. Let us, from now on, identify two mechanisms if one can be
obtained from the other by relabeling commodities. There are three mechanisms
of special interest to us called the \emph{star, cycle, }and \emph{complete
mechanisms; }with the following edge-sets and price ratios:
\[%
\begin{tabular}
[c]{|c|c|c|c|}\hline
$G$ & Star & Cycle & Complete\\\hline
$E_{G}$ & $\left\{  mi,im:i<m\right\}  $ & $\left\{  12,23,\ldots,m1\right\}
$ & $\left\{  ij:i\neq j\right\}  $\\\hline
$p_{i}/p_{j}$ & $b_{mi}b_{jm}/b_{im}b_{mj}$ & $b_{j,j+1}/b_{i,i+1}$ & $\ast
$\\\hline
\end{tabular}
\ \ \ \ \ \ \ \ \ \
\]
For the star and cycle mechanisms, the right-hand side of (\ref{=priceformula}%
) involves a \emph{single} tree and, in the ratio $p_{i}/p_{j},$ several
factors cancel leading to the simple expressions in the table above. However,
for the complete mechanism there is no cancellation and in fact here each
price ratio depends on \textit{every} $b_{ij}.$

The class of $G$\emph{-mechanisms} is the set
\begin{equation}
\mathfrak{M}(m)=\left\{  M_{G}:G\text{ is a directed, connected graph on
}\left\{  1,\ldots,m\right\}  \right\}  .
\end{equation}
Although finite, $\mathfrak{M}(m)$ is rather large, indeed super-exponential
in $m.$ We shall see that some natural complexity considerations help cut down
its size.

Consider a trader who interfaces with $M\in$ $\mathfrak{M}(m)$ in order to
exchange $i$ for $j$. A natural concern for him would be: what is the minimum
number of time periods $\tau_{ij}\left(  M\right)  $ needed to accomplish this
exchange? We define the \emph{time-complexity }of $M$ to be
\begin{equation}
\tau\left(  M\right)  =\max_{i\neq j}\tau_{ij}\left(  M\right)  .
\end{equation}
It is evident that $\tau_{ij}\left(  M\right)  $ is the length of the shortest
path in $G$ from $i$ to $j$ and $\tau\left(  M\right)  $ is the
\textit{diameter} of the graph $G.$

The other concern of our trader would be: how much of commodity $j$ can he get
per unit of $i$? It follows from equation (\ref{MGr}) that he can calculate
this from the state $b$ of the mechanism which determines the\ \emph{price
ratio}\footnote{If there is a continuum of traders (see Section
\ref{continuum}), his own action has no affect on the price ratio. Otherwise
it affects the aggregate offer and thereby the price ratio, which is but to be
expected in an oligopolistic framework. In \emph{either} case, equation
(\ref{MGr}) applies; and $p_{i}/p_{j}$ is the exchange ratio between $i$ and
$j$.} $p_{i}/p_{j}$. Thus the question can be rephrased: how many components
of $b$ does he need to know\footnote{And, since he always knows his own offer,
this is the same as asking: how many components does he need to know of the
aggregate offer of the \emph{others} ?} in order to calculate $p_{i}/p_{j}$?
The table above indicates that it is easier to compute $p_{i}/p_{j}$ for the
star and cycle mechanisms than, say, the complete mechanism.

To make this notion precise, if $f$ is a function of several variables
$x=\left(  x_{1,}\ldots,x_{l}\right)  $, let us say that the component $i$ of
$x$ is \emph{influential} if there are two inputs $x,x^{\prime}$, differing
only in the $i$-th place, such that $f\left(  x\right)  \neq f\left(
x^{\prime}\right)  $. Define\emph{ }$\pi_{ij}(M\mathcal{)}$ to be the number
of influential components of $b$ in the price ratio function $p_{i}/p_{j}.$
For example, from the expression for $p_{i}/p_{j}$ for the star mechanism in
the previous table, it is clear that $\pi_{ij}(M\mathcal{)}$ is $4$ unless one
of $i$ or $j$ is $m,$ in which case it is $2$. We define the\emph{ price
complexity }of $M$ to be%
\begin{equation}
\pi(M\mathcal{)}=\max_{i\neq j}\pi_{ij}(M\mathcal{)}.
\end{equation}

We now define a \emph{quasiorder} $\preceq$ (reflexive and transitive) on
$\mathfrak{M}(m)$ by%
\begin{equation}
M\preceq M^{\prime}\text{\quad}\iff\quad\tau(M)\leq\tau^{\prime}(M^{\prime
})\text{ \emph{and }}\pi(M)\leq\pi^{\prime}\left(  M^{\prime}\right)
\end{equation}

We are ready to state our main result\footnote{A word about the numbering
system used in this paper: all theorems, remarks, conditions, lemmas etc. are
arranged in a \emph{single }grand sequence. Thus the reader shall see, in
order of appearance: Theorem 1, Theorem 2, Remark 3, Condition 4,\ldots. This
does \emph{not }mean that Condition 4 is the fourth condition; in fact it is
the first condition, but it has fourth place in the grand sequence (and, the
marker 4 makes the condition easy to locate).}.

\begin{theorem}
\label{Emergence of money}If\footnote{When $m=3$, we get a fourth mechanism
with complexities $4,2$ identical to the star mechanism. And when $m=2$, we
must change $4$ to $2$ in the table (the three graphs become identical with
complexities $2,2$ for each).} $m>3$ then the three special mechanisms are
\emph{precisely} the $\preceq$-minimal\footnote{$M$ is said to be $\preceq
$\emph{-minimal }in\emph{ }$\mathfrak{M}(m)$ if there is no $M^{\prime}%
\in\mathfrak{M}(m)$ for which $\tau(M^{\prime})\leq\tau(M)$ and $\pi
(M^{\prime})\leq\pi(M)$, with strict inequality in at least one place.}
elements of $\mathfrak{M}(m)$. Their complexities are as follows\emph{:}%
\[%
\begin{tabular}
[c]{|c|c|c|c|}\hline
& Star & Cycle & Complete\\\hline
$\pi(M\mathcal{)}$ & $4$ & $2$ & $m(m-1)$\\\hline
$\tau(M)$ & $2$ & $m-1$ & $1$\\\hline
\end{tabular}
\ \ \
\]

\end{theorem}

This has the following immediate consequence.

\begin{theorem}
\label{Uniqueness of money} Given any choice of strictly positive weights
$\lambda,\mu>0,$ there exists an integer $m_{0}$ such that for $m\geq m_{0}$
the star mechanism is the \emph{unique }minimizer in $\mathfrak{M}(m)$ of
$\lambda\pi(M)+\mu\tau(M\mathcal{)}.$
\end{theorem}

Theorem \ref{Uniqueness of money} says that, so long as traders ascribe
positive weight to \emph{both} time and price complexity considerations, the
star mechanism with money is the unique optimal mechanism as soon as the
number of commodities is sufficiently large.

\begin{remark}
In fact $m_{0}$ does not have to be too large. We only require $4\lambda
+2\mu<2\lambda+(m-1)\mu$ and $4\lambda+2\mu<m(m-1)\lambda+\mu$ for the star to
beat the cycle and complete mechanisms, respectively; which may be rearranged%
\[
m>2\left(  \frac{\lambda}{\mu}\right)  +3\text{ and }m^{2}-m>\frac{\mu
}{\lambda}+4
\]
So, for example, if at least $10\%$ weight is accorded to both $\pi$ and
$\mu,$ then $\lambda/\mu$ and $\mu/\lambda$ can each be at most $9$ and the
above inequalities will hold if $m>18+3$ and $m^{2}-m>9+4$; thus $m_{0}=22$
does the job.
\end{remark}

\begin{remark}
\label{inside complexity} Our notion $\pi_{ij}(M)$ of price complexity counts
the number of components of the market state $b$ that are needed for the
computation of price ratios $p_{i}/p_{j}$. The difficulty of that computation
is not taken into account. However, even if it were, the star would perform
well relative to the other $G$-mechanisms. The intuition for this is implicit
in our earlier discussion. First recall that, by equation (\ref{MGr}), the
returns to the traders are immediate from the price ratios $p_{i}/p_{j}.$ As
for these ratios, they are given by equation (\ref{=priceformula}), which
entails spanning trees of the graph $G.$ For most graphs $G$ this leads to
complicated expressions for $p_{i}/p_{j}$. But, as was said, the star has
unique spanning trees for each commodity, with heavy overlaps between them.
This enables cancellations in the right-hand of equation \ref{MGr}, yielding
the simple formula $p_{i}/p_{j}=b_{mi}b_{jm}/b_{im}b_{mj}$ (see the first
table in Section \ref{Formal model}) whose \textquotedblleft computational
complexity\textquotedblright\ is hardly worth the mention.
\end{remark}

\begin{remark}
\label{clearinghouse} If \textquotedblleft netting\textquotedblright\ of
commodities were permitted, an individual could trade $i$ for $j$ in one go,
instead of trading iteratively along the path that connects $i$ to $j.$ This
reduction of time complexity is quite illusory, however. It requires
bookkeeping --- carried out by a centralized clearing-house ? --- to determine
the net due to, or owed by, any individual across his many trades. Thus
netting simply transfers time complexity to the complexity of bookkeeping.
There is another complication with netting. What if someone is unable to honor
his net debts? Are his final holdings to be confiscated? How, and at what
cost? And furthermore how are the confiscated goods to be apportioned among
the many claimants? Without being formal about it, it should be intuitively
clear that such a clearing-house is complicated and costly to operate, and yet
it is unavoidable if \textquotedblleft netting\textquotedblright\ is to be
accomodated; or, for that matter, borrowing and default, or other variations
of the (\textquotedblleft value-for-value\textquotedblright\ and
\textquotedblleft on-the-spot\textquotedblright) trade that prevails in our
G-mechanisms. More importantly, note that such variations cannot even be
defined except in the context of some \emph{given }basic\emph{ }mechanism. Our
analysis pertains to basic mechanisms, which are prerequisite to the variations.
\end{remark}

\section{Characterization of G-mechanisms \label{characterization}}

Our analysis above was carried out on the domain $\mathfrak{M}(m)$. We now
show how to derive $\mathfrak{M}(m)$ from a more general standpoint. To this
end, let us first define an \emph{abstract }exchange mechanism on commodity
set $\left\{  1,\ldots,m\right\}  $ and with trading opportunities given by a
directed, connected graph $G$ on $\left\{  1,\ldots,m\right\}  $. Such a
mechanism allows individuals in $\left\{  1,\ldots,n\right\}  $ to trade by
means of quantity offers in each commodity $i$ across all edges $ij$ in $G.$
(Here $m$ is fixed and $n$ can be arbitrary.) The offer of any trader can thus
be viewed as an $m\times m$ non-negative matrix in the space%
\[
S=\left\{  a:a_{ij}=0\text{ if }ij\notin G\text{, }a_{ij}\geq0\text{
otherwise}\right\}
\]
Define
\[
S_{+}=\left\{  a\in S:a_{ij}>0\text{ if }ij\in G\right\}
\]
Also define
\[
\overline{a}=\left(  \overline{a_{1}},\ldots,\overline{a_{m}}\right)
\]
where $\overline{a_{i}}=\sum_{j}a_{ij}$ is the $i$-th row sum of $a$ and
denotes the total amount of commodity $i$ involved in sending offer $a_{i}.$
Let $S^{n}$ be the $n$-fold Cartesian product of $S$ with itself, and (with
$\boldsymbol{a=}$ $(\boldsymbol{a}^{1},\ldots,\boldsymbol{a}^{n})$) let
\[
S(n)=\left\{  \boldsymbol{a}\in S^{n}:\sum_{\alpha=1}^{n}\boldsymbol{a}%
^{\alpha}\in S_{+}\right\}
\]
denote the $n$-tuples of offers that are positive on aggregate. Also let
$C=\mathbb{R}_{+}^{m}$ denote the \emph{commodity space}; and $C^{n}$ its
$n$-fold product.

An \emph{exchange mechanism} $M$, for a given set $\left\{  1,\ldots
,m\right\}  $ of commodities and with trading opportunities in accordance with
the graph $G$, is a collection of maps (one for each positive integer $n$)
from $S(n)$ to $C^{n}$ such that, if\textbf{ }$\boldsymbol{a\in}$ $S(n)$ leads
to returns $\mathbf{r}$ $\in C^{n}$, then we have%
\[
\sum_{\alpha=1}^{n}\overline{\boldsymbol{a}}^{\alpha}=\sum_{\alpha=1}%
^{n}\boldsymbol{r}^{\alpha},
\]
i.e., there is \emph{conservation of commodities}. It is furthermore
understood, in keeping with our concept of opportunity $ij,$ that for an offer
$a\in S$ whose only non-zero components are $\left\{  a_{ij}:j=\ldots\right\}
$, the return will consist exclusively of commodity $j.$

We shall impose four conditions on the mechanisms which reflect
\textquotedblleft convenience\textquotedblright\ and \textquotedblleft
fairness\textquotedblright\ in trade. The first condition is that the
mechanism must be blind to all other characteristics of a trader except for
his offer (and rules out discrimination on irrelevant grounds):

\begin{condition}
[Anonymity]\label{Anonymity} Let $(\boldsymbol{r}^{1},\ldots,\boldsymbol{r}%
^{n})\in C^{n}$ denote the returns from $(\boldsymbol{a}^{1},\ldots
,\boldsymbol{a}^{n})\in S(n)$\textbf{. }Then for any permutation $\sigma$ the
returns from $(\boldsymbol{a}^{\sigma(1)},\ldots,\boldsymbol{a}^{\sigma(n)})$
are $(\boldsymbol{r}^{\sigma(1)},\ldots,\boldsymbol{r}^{\sigma(n)}).$
\end{condition}

The second condition is that if any trader pretends to be two different
persons by splitting his offer, the returns to the others is unaffected. In
its absence, traders would be faced with the complicated task of tracking
everyone's offers. It is easier (and sufficient!) to state this condition for
the \textquotedblleft last\textquotedblright\ trader.

\begin{condition}
[Aggregation]\label{Aggregation} Suppose $\boldsymbol{a}\in S(n)$ and
$\boldsymbol{b}$ $\in S(n+1)$ are such that $\boldsymbol{a}^{\alpha
}=\boldsymbol{b}^{\alpha}$ for $\alpha<n$ and $\boldsymbol{a}^{n}%
=\boldsymbol{b}^{n}+\boldsymbol{b}^{n+1}$ . Let $\mathbf{r}$,$\mathbf{s}$
denote the returns that accrue from $\boldsymbol{a}\mathbf{,}\boldsymbol{b}$
respectively. Then $\mathbf{r}^{\alpha}=\boldsymbol{s}^{\alpha}$ for
$\alpha<n.$
\end{condition}

\emph{Anonymity }and\emph{ Aggregation} immediately imply that, regardless of
the size $n$ of the population, the return to any trader may be written
$r(a,b),$where $a\in S$ is his own offer and $b\in S_{+}$ is the aggregate of
all offers.

Let $\nu$ denote his \emph{net trade}:%
\[
\nu(a,b)=r(a,b)-\overline{a}%
\]

The third condition is \emph{Invariance. }Its main content is that the
\textit{maps} which comprise $M$ are invariant under a change of units in
which commodities are measured. This makes the mechanism much simpler to
operate in: one does not need to keep track of seven pounds or seven kilograms
or seven tons, just the numeral $7$ will do.

In what follows, we will consistently use $a$ for an individual's offer and
$b$ for the positive aggregate offer; so, when we refer to the pair $a,b$ it
will be implicit that $a\in S$, $b\in S_{+}$ and $a\leq b$.

\begin{condition}
[Invariance]\label{Invariance}\ $\nu(\lambda a,\lambda b)=$ $\lambda\nu(a,b)$
for all $a,b$ and any $m\times m$ strictly positive diagonal matrix $\lambda.$
\end{condition}

The fourth, and last, condition is that no trader can get strictly less than
his offer (otherwise, such unfortunate traders would tend to abandon the mechanism).

\begin{condition}
[Non-dissipation]\label{Non-dissipation} If $\nu(a,b)\neq0,$ then $\nu
_{i}(a,b)>0$ for some component $i.$
\end{condition}

It turns out that these four conditions categorically determine a unique mechanism.

\begin{theorem}
\label{characterization theorem}Let $M$ be an exchange mechanism on commodity
set $\left\{  1,\ldots,m\right\}  $ and let $G$ be the (directed, connected)
graph induced by the trading opportunities in $M.$ If $M$ satisfies
\emph{Anonymity}, \emph{Aggregation}, \emph{Invariance} and
\emph{Non-dissipation}, then $M=M_{G}.$
\end{theorem}

\subsubsection{Comments on the Conditions \label{no arbitrage}}

\emph{Aggregation} does not imply that if two individuals were to merge, they
would be unable to enhance their \textquotedblleft
oligopolistic\ power\textquotedblright. For despite the \emph{Aggregation}
condition, the merged individuals are free to \emph{coordinate} their actions
by jointly picking a point in the Cartesian product of their action spaces.
Indeed all the mechanisms we obtain display this \textquotedblleft
oligopolistic effect\textquotedblright, even though they also satisfy
\emph{Aggregation}.

It is worthy of note that the cuneiform tablets of ancient Sumeria, which are
some of the earliest examples of written language\ and arithmetic, are in
large part devoted to records and receipts pertaining to economic
transactions. \emph{Invariance} postulates the "numericity" property of the
maps $r(a,b)$ (equivalently, $\nu(a,b)$) making them independent of the
underlying choice of units, and this goes to the very heart of the
quantitative measurement of commodities. In its absence, one would need to
figure out how the maps are altered when units change, as they are prone to
do, especially in a dynamic economy. This would make the mechanism cumbersome
to use.

\emph{Non-dissipation }(in conjunction with \emph{Aggregation, Anonymity}, and
the conservation of commodities) immediately implies \emph{no-arbitrage: }for
any $a,b$ neither $\nu(a,b)\gvertneqq0$ nor $\nu(a,b)\lvertneqq0.$ To check
this, we need consider only the case $a\leq b$ and rule out $\nu
(a,b)\gvertneqq0.$ Denote $c=b-a.$ Then $\nu(a,b)+\nu(c,b)=\nu(a+c,b)=\nu
(b,b)=0,$ where the first equality follows from \emph{Aggregation}, and the
last from conservation of commodities. But then $\nu(a,b)\gvertneqq0$ implies
$\nu(c,b)\lvertneqq0,$ contradicting \emph{Non-dissipation.}

\begin{remark}
\label{touch of rationality} There is a \textquotedblleft
touch\textquotedblright\ of rationality, imputed to the traders, in these
conditions. \emph{Non-dissipation} implies that commodities are liked and an
uncompensated loss of them is not tolerable. (This is compatible with
\emph{any} monotonic utility function and hardly very restrictive.)
\emph{Anonymity }rules out discrimination among traders on extra-economic
grounds. \emph{Aggregation} and \emph{Invariance, }as well our notion of the
complexity of a mechanism, reflect the fact that traders find complicated
computations inconvenient. These requirements are minimalistic and an
order-of-magnitude milder than the standard utilitarian (or other behavioral)
considerations. In fact, our mechanisms permit \emph{arbitrary }utility
functions to be ascribed to the traders in order to build a game (see, e.g.,
\cite{Shapley: 1976}, \cite{Dubey-Shubik:1978} \cite{Sahi-Yao:1989} and the
references therein).
\end{remark}

\subsubsection{Alternative Characterizations of G-Mechanisms}

The formula (\ref{MGr}) for the return function of a $G$-mechanism immediately
implies
\begin{equation}
p(b)=p(c)\Longrightarrow r(a,b)=r(a,c)\text{ for all }a\geq0\text{ and }b,c>0
\label{PrMed}%
\end{equation}
In \cite{Dubey-Sahi: 2003}, a mechanism was supposed to produce both trades
and prices, based upon everyone's offers; and the property (\ref{PrMed}) was
referred to as \emph{Price Mediation}. It was shown in \cite{Dubey-Sahi: 2003}
that $\mathfrak{M}(m)$ is characterized by \emph{Anonymity, Aggregation,
Invariance, Price Mediation }and \emph{Accessibility }(the last representing a
weak form of continuity)$.$

An alternative characterization of $\mathfrak{M}(m)$, which assumes -- as we
do here -- that a mechanism produces only trades (and no prices), was given in
\cite{Dubey-Sahi-Shubik: 2014}. Here we have presented a simplified version of
the analysis in \cite{Dubey-Sahi-Shubik: 2014}, and established that $M_{G}$
arises \textquotedblleft naturally\textquotedblright\ \emph{once }we assume
that trading opportunities are restricted to pairwise exchange of commodities,
i.e., correspond to the edges of a connected graph $G.$ In contrast, in both
\cite{Dubey-Sahi: 2003} and \cite{Dubey-Sahi-Shubik: 2014}, the opportunity
structure $G$ was \textit{itself} an object of deduction, starting from a more
abstract viewpoint; and is the theme of our forthcoming companion paper
\cite{Dubey-Sahi-Shubik:2015}, (extracted from \cite{Dubey-Sahi-Shubik: 2014}).

\section{Proofs}

\subsection{Graphs with complexity $\leq4$}

Let $G$ be a connected graph on $\left\{  1,\ldots,m\right\}  $ as in Section
\ref{Formal model}, and write
\[
p_{i}\left(  G\right)  =p_{i}\left(  M_{G}\right)  \text{,\quad}p_{ij}\left(
G\right)  =p_{ij}\left(  M_{G}\right)  \text{ and }\pi\left(  G\right)
=\pi\left(  M_{G}\right)
\]
If $G$ consists of a single vertex then $\pi\left(  G\right)  =0$ by definition.

\begin{lemma}
\label{lm:cycle}If $G$ is a cycle then $\pi\left(  G\right)  =2.$
\end{lemma}

\begin{proof}
Each vertex $i$ in a cycle has a unique outgoing edge, and we denote its
weight by\footnote{This is a departure from our convention heretofore that $a$
shall refer to an individual's offer, and $b$ to the aggregate offer; but
there should be no confusion.} $a_{i}$. For each $i$ we have $p_{i}%
=b_{G}/a_{i}$ where $b_{G}=\prod\nolimits_{ij\in G}b_{ij}=\prod\nolimits_{i}%
a_{i}$ as in (\ref{=priceformula}); hence $p_{i}/p_{j}=a_{j}/a_{i}$ and the
result follows.
\end{proof}

By a \emph{chorded} \emph{cycle} we mean a graph that is a union $G=C\cup P$
where $C$ is a cycle and $P$, the chord, is a path that connects two distinct
vertices of $C$, but which is otherwise disjoint from $C$.

\begin{lemma}
If $G=C\cup P$ is a chorded cycle then $\pi\left(  G\right)  =4$.
\end{lemma}

\begin{proof}
Let $i$ be the initial vertex of the path $P$, then $i$ has two outgoing
edges, $ij$ and $ik$ say, on the cycle and path respectively. Any vertex
$l\neq i$ has a unique outgoing edge, and we denote its weight by $a_{l}$ as
before. Let $x$ be the terminal vertex of the path $P$. If $x=j$ then $G$ has
two $j$-trees, otherwise there is a unique $j$-tree; similarly if $x=k$ then
there are two $k$-trees, otherwise there is a unique $k$-tree. Thus we get the
following table:%
\[%
\begin{tabular}
[c]{|c|c|c|c|}\hline
& $x=j$ & $x=k$ & $x\neq j,k$\\\hline
$p_{j}/b_{G}$ & $a_{j}^{-1}\left(  b_{ik}^{-1}+b_{ij}^{-1}\right)  $ &
$a_{j}^{-1}b_{ik}^{-1}$ & $a_{j}^{-1}b_{ik}^{-1}$\\\hline
$p_{k}/b_{G}$ & $a_{k}^{-1}b_{ij}^{-1}$ & $a_{k}^{-1}\left(  b_{ik}%
^{-1}+b_{ij}^{-1}\right)  $ & $a_{k}^{-1}b_{ij}^{-1}$\\\hline
\end{tabular}
\ \
\]
In every case, the ratio $p_{j}/p_{k}$ depends on all $4$ variables
$a_{j},a_{k},b_{ij},b_{ik}$, thus $\pi\left(  G\right)  \geq4$.

On the other hand, since all vertices other than $i$ have a unique outgoing
edge, it follows that if $x$ is any vertex then every $x$-tree contains all
the outgoing edges except perhaps the edges $b_{ij},b_{ik}$ and $a_{x}$ (if
$x\neq i$); thus $p_{x}$ is divisible by all other weights. It follows that
for any two vertices $x,y$ the ratio $p_{x}/p_{y}$ can only depend on the
variables $b_{ij},b_{ik},a_{x},a_{y}$. Thus we get $\pi\left(  G\right)
\leq4$ and hence $\pi\left(  G\right)  =4$ as desired.
\end{proof}

\begin{remark}
\label{T0}A special case of a chorded cycle is a graph $T_{0}$ with three
vertices that we call a chorded triangle.%
\[%
\begin{tabular}
[c]{|lll|}\hline
$3$ &  & \\
$\uparrow\downarrow$ & $\nwarrow$ & \\
$1$ & $\longrightarrow$ & $2$\\\hline
\end{tabular}
\ \ \qquad%
\begin{tabular}
[c]{|l|l|}\hline
$p_{1}$ & $b_{23}b_{31}$\\\hline
$p_{2}$ & $b_{12}b_{31}$\\\hline
$p_{3}$ & $b_{23}\left(  b_{12}+b_{13}\right)  $\\\hline
\end{tabular}
\ \ \qquad%
\begin{tabular}
[c]{|l|l|}\hline
$p_{1}/p_{2}$ & $b_{23}/b_{12}$\\\hline
$p_{2}/p_{3}$ & $b_{12}b_{31}/b_{23}\left(  b_{12}+b_{13}\right)  $\\\hline
$p_{3}/p_{1}$ & $\left(  b_{12}+b_{13}\right)  /b_{31}$\\\hline
\end{tabular}
\ \
\]

For future use we note that for each index $j$ there is an $i$ such that
$\pi_{ij}\geq3.$
\end{remark}

By a $k$-\emph{rose} we mean a graph that is a union $C_{1}\cup\cdots\cup
C_{k}$, where the $C_{i}$ are cycles that share a single vertex $j$, but which
are otherwise disjoint. Thus a $0$-rose is a single vertex and a $1$-rose is a
cycle. If $G$ is a $k$-rose for some $k\geq2$ then we will simply say that $G$
is a \emph{rose}$.$

If each cycle in a rose $G$ has exactly two vertices, \textit{i.e.,} is a
bidirected edge, then we say that $G$ is a \emph{star}.

\begin{lemma}
If $G$ is a rose then $\pi\left(  G\right)  =4$.
\end{lemma}

\begin{proof}
Let $G$ be the union of cycles $C_{1}\cup\cdots\cup C_{k}$ with common vertex
$j$ as above. Let $a_{1},\ldots,a_{k}$ be the weights of the outgoing edges
from $j$ in cycles $C_{1},\ldots,C_{k}$ respectively, and for all other
vertices $x$ let $b_{x}$ denote the weight of the unique outgoing edge at $x.$
It is easy to see that there for each vertex $v$ of $G$ there is a unique
$v$-tree, and thus the price vectors are given as follows:%
\[
p_{j}=\prod_{x\neq j}b_{x}\text{, \quad}p_{x}=\frac{a_{i}p_{j}}{b_{x}}\text{
if }x\neq j\text{ is a vertex of }C_{i}%
\]
Thus we get%
\[
p_{j}/p_{x}=b_{x}/a_{i},\quad p_{y}/p_{x}=b_{x}a_{l}/b_{y}a_{i}\text{ if
}y\neq j\text{ is a vertex of }C_{l}%
\]
Taking $i\neq l$, we see that $p_{y}/p_{x}$ depends on $4$ variables, and
$\pi\left(  G\right)  =4$.
\end{proof}

Our main result is a classification of connected graphs with $\pi\left(
G\right)  \leq4$.

\begin{theorem}
\label{main}If $G$ is not a chorded cycle or a $k$-rose, then $\pi\left(
G\right)  \geq5$.
\end{theorem}

We give a brief sketch of the proof of this theorem, which will be carried out
in the rest of this section. The actual proof is organized somewhat
differently, but the main ideas are as follows.

We say that a graph $H$ is a \emph{minor} of $G$, if $H$ can be obtained from
$G$ by removing some edges and vertices, and collapsing certain kinds of
edges. Our first key result is that the property $\pi\left(  G\right)  \leq4$
is a \emph{hereditary} property, in the sense that connected minors of such
graphs also satisfy the property. The usual procedure for studying a
hereditary property is to identify the \emph{forbidden minors}, namely a set
$\Gamma$ of graphs such that $G$ fails to have the property iff it contains
one of the graphs from $\Gamma$. We identify a finite collection of such
graphs. The final step is to show that if $G$ is not a chorded cycle or a
$k$-rose then it contains one of the forbidden minors.

We note the following immediate consequence of the results of this section.

\begin{corollary}
\label{not-cycle}If $G$ is not a cycle then $\pi_{ij}\left(  G\right)  \geq4$
for some $ij$.
\end{corollary}

\subsection{Subgraphs}

Throughout this section $G$ denotes a connected graph. We say that a graph $H$
is a \emph{subgraph} of $G$ if $H$ is obtained from $G$ by deleting some edges
and vertices.

\begin{proposition}
\label{subgraph}If $G^{\prime}$ is a connected subgraph of $G$ then
$\pi\left(  G\right)  \geq\pi\left(  G^{\prime}\right)  $.
\end{proposition}

\begin{proof}
For a vertex $i$ in $G^{\prime}$ let $p_{i}^{\prime}$ and $p_{i}$ denote its
price in $G^{\prime}$ and $G$ respectively; we first relate $p_{i}^{\prime}$
to a certain specialization of $p_{i}$.

Let $E,E^{\prime}$ be the edge sets of $G,G^{\prime}$ respectively, and let
$E_{0}$ (resp. $E_{1}$) denote the edges in $E\setminus E^{\prime}$ whose
source vertex is inside (resp. outside) $G^{\prime}$. Let $\bar{p}_{i}$ be the
specialization of $p_{i}$ obtained by setting the edge weights in $E_{0}$ and
$E_{1}$ to $0$ and $1$ respectively. Then we claim that%
\begin{equation}
p_{i}^{\prime}=\left\vert F\right\vert \bar{p}_{i}, \label{=pip}%
\end{equation}
where $F$ is the set of directed forests $\phi$ in $G$ such that

\begin{enumerate}
\item the root vertices of $\phi$ are contained in $G^{\prime}$,

\item the non-root vertices of $\phi$ consist of \emph{all} $G$-vertices not
in $G^{\prime}.$
\end{enumerate}

Indeed, consider the expression of $p_{i}$ as a sum of $i$-trees in $G$. The
specialization $\bar{p}_{i}$ assigns zero weight to all trees with an edge
from $E_{0}$. The remaining $i$-trees in $G$ are precisely of the from
$\tau\cup\phi$ where $\tau$ is an $i$-tree in $G^{\prime}$ and $\phi\in F$,
and these get assigned weight $wt\left(  \tau\right)  $. Formula (\ref{=pip})
is an immediate consequence.

Now if $i,j$ are vertices in $G^{\prime}$, then formula (\ref{=pip}) gives%
\[
\frac{p_{i}^{\prime}}{p_{j}^{\prime}}=\frac{\bar{p}_{i}}{\bar{p}_{j}}%
\]
Thus the $ij$ price ratio in $G^{\prime}$ is obtained by a
\emph{specialization} of the ratio in $G$. Consequently the former cannot
involve \emph{more} variables. Taking the maximum over all $i,j$ we get
$\pi\left(  G\right)  \geq\pi\left(  G^{\prime}\right)  $ as desired.
\end{proof}

\subsection{Collapsible edges}

We write \emph{out}$\left(  k\right)  $ for the number of outgoing edges at
the vertex $k$. In a connected graph we have \emph{out}$\left(  k\right)
\geq1$ for all vertices, and we will say $k$ is \emph{ordinary} if
\emph{out}$\left(  k\right)  =1$ and \emph{special }if \emph{out}$\left(
k\right)  >1$. Among special vertices, we will say that $k$ is \emph{binary}
if \emph{out}$\left(  k\right)  =2$ and \emph{tertiary} if \emph{out}$\left(
k\right)  =3$.

\begin{definition}
We say that an edge $ij$ of a graph $G$ is \emph{collapsible} if

\begin{enumerate}
\item $i$ is an ordinary vertex

\item $ji$ is not an edge of $G$

\item there is no vertex $k$ such that $ki$ and $kj$ are both edges of $G.$
\end{enumerate}
\end{definition}

\begin{definition}
If $G$ has no collapsible edges we will say $G$ is \emph{rigid}.
\end{definition}

If $G$ is a connected graph with a collapsible edge $ij$, we define the
$ij$-\emph{collapse} of $G$ to be the graph $G^{\prime}$ obtained by deleting
the vertex $i$ and the edge $ij$, and replacing any edges of the form $li$
with edges $lj$. The assumptions on $ij$ imply that the procedure does not
introduce any loops or double edges, hence $G^{\prime}$ is also simple (and
connected). Moreover each vertex $k\neq i$ has the same outdegree in
$G^{\prime}$ as in $G.$

\begin{lemma}
\label{collapse}If $G^{\prime}$ is the $ij$-collapse of $G$ as above, then
$\pi\left(  G\right)  \geq\pi\left(  G^{\prime}\right)  .$
\end{lemma}

\begin{proof}
Let $k$ be any vertex of $G^{\prime}$ then $k$ is also a vertex of $G$. Since
$i$ is ordinary every $k$-tree in $G$ must contain the edge $ij$; collapsing
this edge gives a $k$-tree in $G^{\prime}$ and moreover every $k$-tree in
$G^{\prime}$ arises uniquely in this manner. Thus we have a factorization%
\[
p_{k}\left(  G\right)  =a_{ij}p_{k}\left(  G^{\prime}\right)  .
\]
Thus for any two vertices $k,l$ of $G^{\prime}$ we get $p_{k}\left(  G\right)
/p_{l}\left(  G\right)  =p_{k}\left(  G^{\prime}\right)  /p_{l}\left(
G^{\prime}\right)  $ and the result follows.
\end{proof}

We will say that $H$ is a\emph{ minor} of $G$ if it is obtained from $G$ by a
\emph{sequence} of steps of the following kind: a) passing to a connected
subgraph, b) collapsing some collapsible edges. By Proposition \ref{subgraph}
and Lemma \ref{collapse} we get

\begin{corollary}
\label{minor}If $H$ is a minor of $G$ then $\pi\left(  H\right)  \leq
\pi\left(  G\right)  .$
\end{corollary}

\subsection{Augmentation}

Throughout this section $G$ denotes a connected graph.

\begin{notation}
We write $H\trianglelefteq G$ if $H$ is a \emph{connected} subgraph of $G$,
and write $H\vartriangleleft G$ to mean $H\trianglelefteq G$ and $H\neq G$.
\end{notation}

We say that $H\vartriangleleft G$ can be \emph{augmented} if there is a path
$P$ in $G$ whose endpoints are in $H$, but which is otherwise completely
disjoint from $H$. We refer to $P$ as an augmenting path of $H$, and to
$K=H\cup P$ as an augmented graph of $H$; note that $K$ is also connected,
\textit{i.e.} $K\trianglelefteq G$. It turns out that augmentation is always possible.

\begin{lemma}
If $H\vartriangleleft G$ then $H$ can be augmented.
\end{lemma}

\begin{proof}
If $G$ and $H$ have the same vertex set then any edge in $G\setminus H$
comprises an augmenting path. Otherwise consider triples $\left(
k,P_{1},P_{2}\right)  $ where $k$ is a vertex not in $H$, $P_{1}$ is a path
from some vertex in $H$ to $k$, and $P_{2}$ is a path from $k$ to some vertex
in $H$. Among all such triples choose one with $e\left(  P_{1}\right)
+e\left(  P_{2}\right)  $ as small as possible. Then $P_{1}$ and $P_{2}$
cannot share any \emph{intermediate} vertices with $H$ or with each other,
else we could construct a smaller triple. It follows that $P=P_{1}\cup P_{2}$
is an augmenting path.
\end{proof}

We are particularly interested in augmenting paths for $H$ that consist of one
or two edges; we refer to these as \emph{short augmentations} of $H$.

\begin{corollary}
\label{augmentation}If $H\vartriangleleft G$ then $G$ has a minor that is a
short augmentation of $H$.
\end{corollary}

\begin{proof}
Let $K=H\cup P$ be an augmentation of $H$. If $P$ has more than two edges,
then we may collapse the first edge of $P$ in $K$. The resulting graph is a
minor of $G,$ which is again an augmentation of $H$. The result follows by iteration.
\end{proof}

\begin{lemma}
\label{complexity jump}If $K=H\cup P$ with $P=\left\{  jk,kl\right\}  $, then
for any vertex $i$ of $H$ we have $\pi_{ik}\left(  K\right)  =\pi_{ij}\left(
H\right)  +2.$
\end{lemma}

\begin{proof}
The edges $\left(  j,k\right)  $ and $\left(  k,l\right)  $ are the unique
incoming and outgoing edges at $k$. It follows that every $i$-tree in $K$ is
obtained by adding the edge $kl$ to an $i$-tree in $H$, and every $k$-tree in
$K$ is obtained by adding the edge $jk$ to a $j$-tree in $H$. Thus if $a_{jk}$
and $a_{kl}$ are the respective weights of the two edges in the path $P$ then
we have
\[
p_{i}\left(  K\right)  =a_{kl}p_{i}\left(  H\right)  ,p_{k}\left(  K\right)
=a_{jk}p_{j}\left(  H\right)  \implies\frac{p_{i}\left(  K\right)  }%
{p_{k}\left(  K\right)  }=\frac{a_{kl}}{a_{jk}}\frac{p_{i}\left(  H\right)
}{p_{j}\left(  H\right)  }%
\]
Thus the price ratio in question depends on two additional variables, and the
result follows.
\end{proof}

\begin{corollary}
\label{CCE}If $G$ contains the chorded triangle $T_{0}$ as a proper subgraph
then $\pi\left(  G\right)  \geq5.$
\end{corollary}

\begin{proof}
By Corollary \ref{augmentation}, $G$ has a minor $K=T_{0}\cup P$, which is a
short augmentation of $T_{0}$, and it is enough to show that $\pi\left(
K\right)  \geq5$. If $P$ consists of two edges $\left\{  jk,kl\right\}  $ then
by Remark \ref{T0} we can choose $i$ such that $\pi_{ij}\left(  T_{0}\right)
=3$; now by Lemma \ref{complexity jump}, we have $c_{ik}\left(  K\right)  =5$
and hence $\pi\left(  K\right)  \geq5$. If $P$ consists of a single edge then
$K$ is necessarily as below, and once again $\pi\left(  K\right)  \geq5$.%
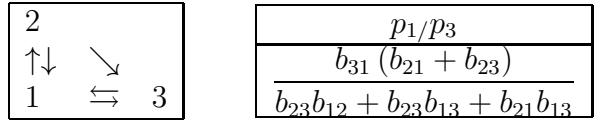
\[%
\begin{tabular}
[c]{|lll|}\hline
$2$ &  & \\
$\uparrow\downarrow$ & $\searrow$ & \\
$1$ & $\leftrightarrows$ & $3$\\\hline
\end{tabular}
\ \ \ \ \ \ \
\begin{tabular}
[c]{|c|}\hline
$p_{1/}p_{3}$\\\hline
$\dfrac{b_{31}\left(  b_{21}+b_{23}\right)  }{b_{23}b_{12}+b_{23}b_{13}%
+b_{21}b_{13}}$\\\hline
\end{tabular}
\ \ \ \ \
\]

\end{proof}

\subsection{The circuit rank}

As usual $G$ denotes a simple connected graph, and we will write $e\left(
G\right)  $ and $v\left(  G\right)  $ for the numbers of edges and vertices of
$G$.

\begin{definition}
The circuit rank of $G$ is defined to be
\[
c\left(  G\right)  =e\left(  G\right)  -v\left(  G\right)  +1
\]

\end{definition}

The circuit rank is also known as the \emph{cyclomatic number}, and it counts
the number of independent cycles in $G$, see \textit{e.g.} \cite{Berge}.

\begin{example}
If $G$ is a $k$-rose then $c\left(  G\right)  =k$, and if $G$ is a chorded
cycle then $c\left(  G\right)  =2$.
\end{example}

We now prove a crucial property of $c\left(  G\right)  $.

\begin{proposition}
\label{C-rank jump}If $H\vartriangleleft G$ then there is some
$K\trianglelefteq G$ such that $H\vartriangleleft K$ and $c\left(  K\right)
=c\left(  H\right)  +1$.
\end{proposition}

\begin{proof}
Let $K=H\cup P$ be an augmentation of $H$. If $P$ consists of $m$ edges, then
$K$ has $e\left(  H\right)  +m$ edges and $v\left(  H\right)  +$ $m-1$
vertices; hence $c\left(  K\right)  =c\left(  H\right)  +1$.
\end{proof}

\begin{corollary}
\label{C-rank cases}Let $G$ be a connected graph.

\begin{enumerate}
\item If $H\vartriangleleft G$ then $c\left(  H\right)  <c\left(  G\right)  $.

\item $c\left(  G\right)  =0$ iff $G$ is a single vertex.

\item $c\left(  G\right)  =1$ iff $G$ is a cycle.

\item $c\left(  G\right)  =2$ iff $G$ is a chorded cycle or a $2$-rose.
\end{enumerate}
\end{corollary}

\begin{proof}
The first part follows from Proposition \ref{C-rank jump}, the other parts are
completely straightforward.
\end{proof}

\begin{lemma}
\label{nonrose}If $G$ is not a rose and $c\left(  G\right)  >3$, then there is
some $K\vartriangleleft G$ such that $K$ is not a rose and $c\left(  K\right)
=3$.
\end{lemma}

\begin{proof}
Let $R$ be a $k$-rose in $G$ with $c\left(  R\right)  =k$ as large as
possible, then $R\vartriangleleft G$ by assumption. If $c\left(  R\right)
\leq2$ then any $K\vartriangleleft G$ with $c\left(  K\right)  =3$ is not a
rose. Thus we may assume that $c\left(  R\right)  >2,$ and in particular $R$
has a unique special vertex $i$ and at least three loops. Since $R\neq G,$ $R$
can be augmented, and $S=R\cup P$ is an augmentation, then $P$ cannot both
begin and end at $i$, else $R\cup P$ would be a rose, contradicting the
maximality of $R$. Since there are at most two endpoints of $P,$ we can choose
two \emph{distinct} loops $L_{1}$ and $L_{2}$ of $R$, such that $L_{1}\cup
L_{2}$ contains these endpoints of $P$. Then $K=L_{1}\cup L_{2}\cup P$ is the
desired graph.
\end{proof}

\subsection{Covered vertices}

\begin{definition}
Let $i$ be an ordinary vertex of $G$ with outgoing edge $ij$. We say that a
vertex $k$ \emph{covers} $i$, if one of the following holds:

\begin{enumerate}
\item the edges $ki$ and $kj$ belong to $G$

\item $j=k$ and the edge $ki$ belongs to $G$
\end{enumerate}

If there is no such $k$ then we say that $i$ is an \emph{uncovered} vertex.
\end{definition}

We emphasize that the terminology covered/uncovered is only applicable to
ordinary vertices in a graph $G$. The main point of this definition is the
following simple observation.

\begin{remark}
An ordinary vertex is uncovered iff its outgoing edge is collapsible.
\end{remark}

\begin{lemma}
\label{cover}Suppose $G$ is a connected graph .

\begin{enumerate}
\item If $v\left(  G\right)  \geq3$ then an ordinary vertex cannot cover
another vertex.

\item If $v\left(  G\right)  \geq$ $4$ then a binary vertex can cover at most
one vertex.

\item A tertiary vertex can cover at most three vertices.

\item If $G$ is a rigid graph with $c\left(  G\right)  =3$, then $v\left(
G\right)  \leq4$.
\end{enumerate}
\end{lemma}

\begin{proof}
If $k$ is an ordinary vertex covering $i$ then $G$ must contain the edges $ki$
and $ik$. Thus $i$ and $k$ do not have any other outgoing edges, and if $G$
has a third vertex $j$ then there is no path from $k$ or $i$ to $j$, which
contradicts the connectedness of $G$, thereby proving the first statement.

If $k$ is a binary vertex covering the ordinary vertices $i$ and $j$ then $G$
must contain the edges $ki,kj,ij,ji$. The vertices $i,j,k$ cannot have any
other outgoing edges, so a fourth vertex would contradict the connectedness of
$G$ as before. This proves the second statement.

If a vertex $k$ covers $i$ then there must be an edge from $k$ to $i$. Thus if
\emph{out}$\left(  k\right)  =3$ then $k$ can cover at most three vertices.

If $c\left(  G\right)  =3$ then $G$ has either 2 binary vertices or 1 tertiary
vertex, with the remaining vertices being ordinary. If $v\left(  G\right)  >4$
then by previous two paragraphs $G$ would have an uncovered vertex, which is a contradiction.
\end{proof}

\subsection{Proof of Theorem \ref{main}}

\begin{proposition}
\label{c and pi}If $c\left(  G\right)  \geq3$ and $G$ is not a rose, then
$\pi\left(  G\right)  \geq5$.
\end{proposition}

\begin{proof}
By Proposition \ref{subgraph} and Lemma \ref{nonrose} we may assume that
$c\left(  G\right)  =3$. By Lemma \ref{collapse}, we may further assume that
$G$ is rigid, and thus by Lemma \ref{cover} that $v\left(  G\right)  \leq4$.
We now divide the argument into three cases.

First suppose that $G$ contains a $3$-cycle $C$. We claim that at least one of
the edges of $C$ must be a bidirected edge in $G$, so that $G$ properly
contains a chorded triangle $T_{0}$, whence $\pi\left(  G\right)  \geq5$ by
Corollary \ref{CCE}. Indeed if $G$ has no other vertices outside $C$, then $G$
must have $5$ edges and $3$ vertices and the claim is obvious. Thus we may
suppose that there is an outside vertex $l$. We further claim that $C$
contains two vertices $i,j$ such that $i$ covers $j$. Granted this, it is
immediate that $G$ contains either the bidirected edge $ij$ and $ji$, or the
bidirected edge $jk$ and $kj$ where $k$ is the third vertex of $C$. To prove
the \textquotedblleft further\textquotedblright\ claim we note that the
special vertices of $G$ consist of either a) one tertiary vertex, or b) two
binary vertices. In case a) the connectedness of $G$ implies that the tertiary
vertex must be in $C$, and hence it must cover both the ordinary vertices in
$C$. In case b) either $C$ contains both binary vertices, one of which must
cover the unique ordinary vertex of $C$; or $C$ contains one binary vertex,
which must cover one of the two ordinary vertices of $C$.

Next suppose that $G$ does not contain a $3$-cycle, but does contain a
$4$-cycle labeled $1234$, say. Now $G$ has two additional edges, which cannot
be the diagonals $13,31,24,42$, since otherwise $G$ would have a $3$-cycle;
therefore $G$ must have two bidirected edges. The bidirected edges cannot be
adjacent else $G$ would have a collapsible vertex, therefore $G$ must be the
first graph below, which has $\pi\left(  G\right)  \geq5$.
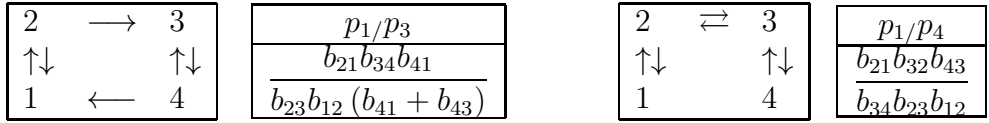
\[
\quad%
\begin{tabular}
[c]{|lll|}\hline
$2$ & $\longrightarrow$ & $3$\\
$\uparrow\downarrow$ &  & $\uparrow\downarrow$\\
$1$ & $\longleftarrow$ & $4$\\\hline
\end{tabular}
\ \ \ \ \
\begin{tabular}
[c]{|c|}\hline
$p_{1/}p_{3}$\\\hline
$\dfrac{b_{21}b_{34}b_{41}}{b_{23}b_{12}\left(  b_{41}+b_{43}\right)  }%
$\\\hline
\end{tabular}
\ \ \ \qquad\quad%
\begin{tabular}
[c]{|lll|}\hline
$2$ & $\rightleftarrows$ & $3$\\
$\uparrow\downarrow$ &  & $\uparrow\downarrow$\\
$1$ &  & $4$\\\hline
\end{tabular}
\ \ \ \
\begin{tabular}
[c]{|c|}\hline
$p_{1/}p_{4}$\\\hline
$\dfrac{b_{21}b_{32}b_{43}}{b_{34}b_{23}b_{12}}$\\\hline
\end{tabular}
\ \ \ \ \
\]

Finally suppose $G$ has no $3$-cycles or $4$-cycles. Then every edge must be a
bidirected edge, and $G$ must be a tree with all bidirected edges. Since $G$
is not a star, this only leaves the second graph above, which has $\pi\left(
G\right)  \geq6.$
\end{proof}

We can now finish the proof of Theorem \ref{main}.

\begin{proof}
[Proof of Theorem \ref{main}]If $c\left(  G\right)  \leq2$ then, by Corollary
\ref{C-rank cases}, $G$ is a single vertex, a cycle, chorded cycle or a
$2$-rose. If $c\left(  G\right)  \geq3$ then the result follows by Proposition
\ref{c and pi}.
\end{proof}

\section{Proof of Theorem \ref{Emergence of money}}

In this section, after a couple of preliminary results, we apply Theorem
\ref{main} to prove Theorem \ref{Emergence of money}.

\begin{lemma}
\label{lem:4-chord}If $G$ is a chorded cycle on 4 or more vertices, then
$\tau\left(  G\right)  \geq3$.
\end{lemma}

\begin{proof}
We can express $G$ as a union of two paths $P,Q$ from $1$ to $2$, say and a
third path $R$ from $2$ to $1$. At least one of the first two paths, say $P$
must have an intermediate vertex, say $3$. Since $m\geq4$ there is an
additional intermediate vertex $4$ on one of the paths.

If $m=4$ then we get three possible graphs depending on the location of the
vertex $4.$%
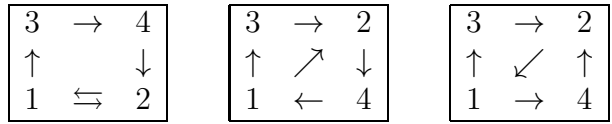
\[%
\begin{tabular}
[c]{|lll|}\hline
$3$ & $\rightarrow$ & $4$\\
$\uparrow$ &  & $\downarrow$\\
$1$ & $\leftrightarrows$ & $2$\\\hline
\end{tabular}
\qquad%
\begin{tabular}
[c]{|lll|}\hline
$3$ & $\rightarrow$ & $2$\\
$\uparrow$ & $\nearrow$ & $\downarrow$\\
$1$ & $\leftarrow$ & $4$\\\hline
\end{tabular}
\qquad%
\begin{tabular}
[c]{|lll|}\hline
$3$ & $\rightarrow$ & $2$\\
$\uparrow$ & $\swarrow$ & $\uparrow$\\
$1$ & $\rightarrow$ & $4$\\\hline
\end{tabular}
\]
For these graphs we have $\tau_{24}=3,\tau_{42}=3$ and $\tau_{34}=3$,
respectively. Thus $\tau\left(  G\right)  \geq3$ in all three cases.

If $m>4$ then $G$ can be realized as one of these graphs, albeit with
additional intermediate vertices on one or more of the paths $P,Q,R$. These
additional vertices are ordinary uncovered vertices, with collapsible outgoing
edges. Collapsing one of these edges does not increase time complexity, and
produces a smaller chorded cycle $G^{\prime}$. Arguing by induction on $m$ we
conclude $\tau\left(  G\right)  \geq\tau\left(  G^{\prime}\right)  \geq3.$
\end{proof}

\begin{lemma}
\label{lem:comp}If $G$ is the complete graph, then $\pi_{ij}\left(  G\right)
=m\left(  m-1\right)  $ for all $i\neq j$.
\end{lemma}

\begin{proof}
Fix a pair of vertices $i\neq j$ in $G$. Then we claim that the price ratio
$p_{ij}\left(  G\right)  $ depends on each of the $m\left(  m-1\right)  $ edge
weights $b_{kl}$. Indeed if $H$ is any "spanning" connected subgraph of $G$
then $p_{ij}\left(  H\right)  $ is obtained from $p_{ij}\left(  G\right)  $ by
specializing to $0$ the weights of all edges outside $H$. Therefore it
suffices to find a connected subgraph $H$ such that $p_{ij}\left(  G\right)  $
depends on $b_{kl}$.

We consider two cases. If $\left\{  i,j\right\}  =\left\{  k,l\right\}  $ then
exchanging $i,j$ if necessary we may assume $i=k,j=l$. Let $H$ be an $m$-cycle
two of whose edges are $ij$ and $hi$ (say); then $p_{i}/p_{j}=b_{hi}/b_{ij}$
depends on $b_{kl}=b_{ij}.$

If $\left\{  i,j\right\}  \neq\left\{  k,l\right\}  $ then let $H$ be an
$2$-rose with loops $C_{1}$ and $C_{2}$ such that

\begin{enumerate}
\item $k$ is the special vertex, and $kl$ is an edge in $C_{1}$

\item $i$ belongs to $C_{1}$ and $j$ belongs to $C_{2}$
\end{enumerate}

Then $p_{i}$ and $p_{j}$ are each given by unique directed trees $T_{i}$ and
$T_{j}$. Moreover $T_{i}$ involves $kl$ while $T_{j}$ does not. Hence
$p_{ij}\left(  H\right)  $ depends on $b_{kl}.$
\end{proof}

\begin{proof}
[Proof of Theorem \ref{Emergence of money}](\textbf{Completion)}Let
$\mathfrak{S}$ denote the set consisting of the three special mechanisms:
star, cycle and complete. We need to show that $\mathfrak{M}_{\preceq
}=\mathfrak{S,}$ where $\mathfrak{M}_{\preceq}$ denotes the set of $\preceq
$-minimal elements of $\mathfrak{M=M(m).}$

Let us say that $G$ is a minimal graph if $M_{G}$ is a minimal mechanism of
$\mathfrak{M}$. Now the star mechanism has complexity $\left(  \tau
,\pi\right)  =\left(  2,4\right)  $. Therefore if $G$ is any minimal graph
then either $\tau\left(  G\right)  =1$ or $\pi\left(  G\right)  \leq4$. For
$\tau\left(  G\right)  =1$ we get the complete graph, which has complexity
$\left(  \tau,\pi\right)  =\left(  1,m\left(  m-1\right)  \right)  $ by Lemma
\ref{lem:comp}. The graphs with $\pi\left(  G\right)  \leq4$ are characterized
by Theorem \ref{main}, and we have three possibilities for $G.$

\begin{enumerate}
\item \textit{Chorded cycle. }In this case we have $\left(  \tau,\pi\right)
=\left(  3^{+},4\right)  $ by Lemma \ref{lem:4-chord}, and so $G$ is not minimal.

\item \textit{Cycle. }In this case we have $\left(  \tau,\pi\right)  =\left(
m-1,2\right)  $ by Lemma \ref{lm:cycle}.

\item $k$\textit{-rose, }$k\geq2$. If each petal of $G$ has exactly 2 edge
then $G$ is the star mechanism. Otherwise after collapsing edges, we obtain
the following minor with $\tau_{12}=3$%
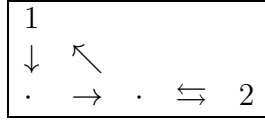
\[%
\begin{tabular}
[c]{|lllll|}\hline
$1$ &  &  &  & \\
$\downarrow$ & $\nwarrow$ &  &  & \\
$\cdot$ & $\rightarrow$ & $\cdot$ & $\leftrightarrows$ & $2$\\\hline
\end{tabular}
\ \
\]
Thus $G$ has complexity $\left(  \tau,\pi\right)  =\left(  3^{+},4\right)  $
and so is not minimal.
\end{enumerate}

Thus the three graphs in the statement of Theorem \ref{Emergence of money} are
the only \emph{possible} minimal graphs, and have the indicated complexities.
Since they are incomparable with each other, each is minimal. Thus we conclude
$\mathfrak{M}_{\preceq}=\mathfrak{S}$ as desired.
\end{proof}

\begin{remark}
For $m=3$, Lemma \ref{lem:4-chord} does not hold and we have an additional
strongly minimal mechanism with $\left(  \tau,\pi\right)  =\left(  2,4\right)
$, namely the chorded triangle
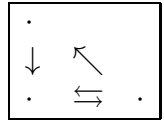
\[%
\begin{tabular}
[c]{|lll|}\hline
$\cdot$ &  & \\
$\downarrow$ & $\nwarrow$ & \\
$\cdot$ & $\leftrightarrows$ & $\cdot$\\\hline
\end{tabular}
\]

\end{remark}

\section{Proof of Theorem \ref{characterization theorem}}

Note that a mechanism is determined uniquely by its \emph{net trade} function
$\nu(a,b):=r(a,b)-\overline{a}$ which, although initially defined for $a\leq
b$, admits a natural extension as follows.

\begin{proposition}
\label{Linearity}The function $\nu$ admits a unique extension to $S\times
S_{+}$ satisfying
\[
\nu(\lambda a+\lambda^{\prime}a^{\prime},b)=\lambda\nu(a,b)+\lambda^{\prime
}\nu(a^{\prime},b),\quad\nu\left(  a,\lambda b\right)  =\nu\left(  a,b\right)
\text{ for }\lambda,\lambda^{\prime}>0
\]

\end{proposition}

\begin{proof}
Since $\nu(a,b):=r(a,b)-\overline{a}$, it suffices to show
\begin{equation}
r(\lambda a+\lambda^{\prime}a^{\prime},b)=\lambda r(a,b)+\lambda^{\prime
}r(a^{\prime},b),\quad r\left(  a,\lambda b\right)  =r\left(  a,b\right)
\text{ for }\lambda,\lambda^{\prime}>0 \label{=rlin}%
\end{equation}

But this is just Lemma 1 of \cite{Dubey-Sahi: 2003}, whose proof we now
reproduce for the sake of completeness.

First observe that, by the conservation of commodities, $r(a,b)\leq
\overline{b}$ for all $a\leq b;$ moreover if $a$ and $a^{\prime}$ in $S$ are
such that $a+a^{\prime}\leq b,$ then \emph{Aggregation }implies the functional
(Cauchy) equation $r(a+a^{\prime},b)=r(a,b)+r(a^{\prime},b)$.

From Corollary 2 in \cite{Aczel-Dhombres:1989} we conclude that, for all
non-negative $\lambda$ and $\lambda^{\prime}$ such that $\lambda
a+\lambda^{\prime}a^{\prime}\leq b,$ the first equality of (\ref{=rlin}) holds.

Next let $a\leq b$ and choose $\lambda\geq1.$ Then the argument just given
shows that $r(\lambda a,\lambda b)=\lambda r(a,\lambda b).$ On the other hand,
\emph{Invariance }implies that the left side equals $\lambda r(a,b).$
Comparing these expressions we obtain the second inequality of (\ref{=rlin}).

Thus even for $a$ \emph{not} less than $b$, we may define $r(a,b)$ via
(\ref{=rlin}) by choosing $\lambda$ sufficiently large. This extends $r$ to
all of $S\times S_{+}.$
\end{proof}

In view of the above result, we drop the restriction $a\leq b$ when
considering $\nu\left(  a,b\right)  $.

The net trade vector can have negative and positive components, and hence
belongs to $\mathbb{R}^{m}$. The next definition pertains to such vectors in
$\mathbb{R}^{m}.$

\begin{definition}
By an $i$-vector, we mean a vector whose $i$th component is positive and all
other components are zero. By an $\bar{\imath}j$-vector we mean a vector that
has a negative $i$-component, a positive $j$-component and zeros in all other components.
\end{definition}

\begin{proposition}
\label{Convertibility}For $b\in S_{+}$ and any $i\neq j$ there is $a\in S$
such that $\nu(a,b)$ is an $\bar{\imath}j$-vector.
\end{proposition}

\begin{proof}
Since the graph $G$ underlying the mechanism is connected, there is a directed
path from $i$ to $j.$ Denote the nodes on the path by $i=1,\ldots,t=j$. Let
$w^{1}$ be an $i$-vector which can be offered on edge $12$ to get a return
$w^{2}\neq0$ consisting only of commodity $2$ (here $w^{2}\neq0$ by
\emph{Non-dissipation}); then $w^{2}$ can be offered on edge $23$ to get
$w^{3}\neq0$ consisting only of commodity $3$, and so on. This yields a
sequence $w^{1},\ldots,w^{t}$ such that
\[
w^{i}+\nu\left(  w^{i},b\right)  =w^{i+1}\text{ for }i=1,\ldots,t-1
\]
If $w=\sum w^{i}$ then by Proposition \ref{Linearity} we have
\[
\nu\left(  w,b\right)  =%
{\textstyle\sum}
\nu\left(  w^{i},b\right)  =w^{t}-w^{1}%
\]
which is an $\bar{\imath}j$-vector.
\end{proof}

It will be convenient to write an $\bar{\imath}j$-vector in the form $\left(
-x,y\right)  $ after suppressing the other components. In the context of the
above proposition if $\nu\left(  a,b\right)  =\left(  -x,y\right)  $ then by
linearity $\nu\left(  a/x,b\right)  =\left(  -1,y/x\right)  $, and we will say
that the offer $a$ (or $a/x$) achieves an $ij$-\emph{exchange ratio} of $y/x$
at $b$.

Proposition \ref{Convertibility} shows that there exists at least one offer
$a$ to achieve an $\bar{\imath}j$-vector in trade, at any given $b$. But $a$
is by no means unique. There may be many paths from $i$ to $j,$ along which
$i$ can be exchanged exclusively for $j;$ and, also, there may be more
complicated trading strategies, that use edges no longer confined to any
single path, to accomplish such an exchange. These could give rise to offers
different from $a$ and yield (for the fixed aggregate $b)$ other $\bar{\imath
}j$-vectors in trade. But, as the following lemma shows, the \emph{same
}exchange ratio obtains under all circumstances.

\begin{lemma}
\label{Common exchange rate}If $a^{\prime},a^{\prime\prime}$achieve
$ij$-exchange ratios $\alpha^{\prime},\alpha^{\prime\prime}$ at $b$, then
$\alpha^{\prime}=\alpha^{\prime\prime}$.
\end{lemma}

\begin{proof}
By Proposition \ref{Convertibility} there exists an $a$ such that $\nu\left(
a,b\right)  $ is a $\bar{j}i$-vector; if $\alpha$ is the corresponding
exchange ratio then by rescaling $a,a^{\prime},a^{\prime\prime}$ we may assume
that
\[
\nu\left(  a,b\right)  =\left(  1,-\alpha\right)  ,\nu\left(  a^{\prime
},b\right)  =\left(  -1,\alpha^{\prime}\right)  ,\nu\left(  a^{\prime\prime
},b\right)  =\left(  -1,\alpha^{\prime\prime}\right)  .
\]
By Proposition \ref{Linearity} we get
\[
\nu\left(  a+a^{\prime},b\right)  =\left(  0,\alpha^{\prime}-\alpha\right)
\]
Now by \emph{Non-dissipation} we get $\alpha\leq\alpha^{\prime}$, and
exchanging the roles of $i$ and $j$ we conclude that $\alpha^{\prime}%
\leq\alpha$ and hence\footnote{Equivalently: \emph{no-arbitrage }of subsection
\ref{no arbitrage} directly implies that $\alpha=\alpha^{\prime}.$} that
$\alpha=\alpha^{\prime}$. Arguing similarly we get $\alpha=\alpha
^{\prime\prime}$ and hence that $\alpha^{\prime}=\alpha^{\prime\prime}$
\end{proof}

\begin{lemma}
\label{budget balance}Denote the net trade function of $M$ by $\nu.$ Then
there is a unique map $p:$ $\mathbb{R}_{++}^{K}\rightarrow$ $\mathbb{R}%
_{++}^{m}/\sim$ satisfying $p(b)\cdot\nu(a,b)=0.$
\end{lemma}

\begin{proof}
Fix $b\in S_{+}$ and consider the vector%
\[
p=\left(  1,p_{2},\ldots,p_{m}\right)
\]
where $p_{j}^{-1}$ is the $1j$-exchange ratio at $b$, as in Lemma
\ref{Common exchange rate}. We will show that $p$ satisfies the budget balance
condition, \textit{i.e.} that
\begin{equation}
p\cdot\nu\left(  a,b\right)  =0\text{ for all }a. \label{=p.nu}%
\end{equation}
We argue by induction on the number $d\left(  a,b\right)  $ of non-zero
components of $\nu\left(  a,b\right)  $ in positions $2,\ldots,m$. If
$d\left(  a,b\right)  =0$ then $\nu\left(  a,b\right)  =0$ by
\emph{Non-dissipation} (enhanced to \emph{no-arbitrage}, see Subsection
\ref{no arbitrage}) and (\ref{=p.nu}) is obvious. If $d\left(  a,b\right)  =1$
then $\nu\left(  a,b\right)  $ is either an $\bar{1}j$-vector or a $\bar{j}1$
vector, which by the definition of $p_{j}$ and Lemma
\ref{Common exchange rate} is necessarily of the form%
\[
\left(  -x,xp_{j}^{-1}\right)  \text{ or }\left(  x,-xp_{j}^{-1}\right)  ;
\]
for such vectors (\ref{=p.nu}) is immediate. Now suppose $d\left(  a,b\right)
=d>1$ and fix $j$ such that $\nu_{j}\left(  a,b\right)  \neq0$. Then we can
choose $a^{\prime}$ such that $\nu\left(  a^{\prime},b\right)  $ is a $\bar
{1}j$ or a $\bar{j}1$- vector such that $\nu_{j}\left(  a,b\right)  =-\nu
_{j}\left(  a^{\prime},b\right)  .$ It follows that $d\left(  a+a^{\prime
},b\right)  <d$ and by linearity we get
\[
p\cdot\nu\left(  a,b\right)  =p\cdot\nu\left(  a+a^{\prime},b\right)
-p\cdot\nu\left(  a^{\prime},b\right)  .
\]
By the inductive hypothesis the right side is zero, hence so is the left side.

Finally the uniqueness of the price function is obvious, because the return
function of the mechanism dictates how many units of $j$ may be obtained for
one unit of $i$, yielding just one possible candidate for the exchange rate
for every pair $ij.$
\end{proof}

We can now prove Theorem \ref{characterization theorem}

\begin{proof}
\textbf{of Theorem} \ref{characterization theorem} (\textbf{Completion)} To
prove that $M=M_{G}$ it is enough to show that $p$ and $r$ satisfy
(\ref{=price}) and (\ref{MGr}).

Let us write, as before,
\[
b=\sum a^{\alpha},p=p(b)\text{ and }\nu\left(  a,b\right)  =r(a,b)-\overline
{a}.
\]
Consider replacing trader $\alpha$ by $m$ traders $\alpha_{1},\ldots
,\alpha_{m},$ where trader $\alpha_{j}$ makes only the offers $\left\{
a_{ij}^{\alpha}:1\leq i\leq m\right\}  $ in $a^{\alpha}$ that entitle $\alpha$
to the return of commodity $j.$ By \emph{Aggregation }this will have no effect
on traders other than $\alpha;$ and hence $\alpha_{j}$ will get precisely the
return $r_{j}(a^{\alpha},b).$ By Lemma \ref{budget balance}, applied to each
such trader $\alpha_{j},$ we have
\begin{equation}
p_{j}r_{j}(a^{\alpha},b)=\sum_{i}p_{i}a_{ij}^{\alpha} \label{=bbal}%
\end{equation}
which is just (\ref{MGr}).

Now (\ref{=price}) follows by summing (\ref{=bbal}) over all $\alpha.$
\end{proof}

\section{ A Continuum of Traders\label{continuum}}

Our analysis easily extends to the case where the set of individuals $T$ is
the unit interval $\left[  0,1\right]  $, endowed with a nonatomic population
measure \footnote{Denote the measure $\mu.$ And since $\mu$ is to be held
fixed throughout, we may suppress it, abbreviating $\int_{T}\mathbf{f}%
$\textbf{ }$(t)d\mu(t)$ by $\int_{T}\mathbf{f}$ \ for any measurable
function\textbf{ }$\mathbf{f}$\textbf{ }on $\left[  0,1\right]  .$}. Let
$\mathcal{S}$ denote the collection of all integrable functions
$\boldsymbol{a}:T\mapsto S$ such that $\int_{T}\boldsymbol{a\in}$ $S_{+}$. (An
element of $\mathcal{S}$ represents a choice of offers by the traders in $T$
which are positive on aggregate.) In the same vein, let $\mathcal{R}$ denote
the collection of all integrable functions from $T$ to $C,$ whose elements
$\mathbf{r}:T\mapsto C$ represent returns to $T.$ An \emph{exchange mechanism}
$M$, on a given set of $m$ commodities, is a map from $\mathcal{S}$ to
$\mathcal{R}$ such that, if $M$ maps\textbf{ }$\boldsymbol{a}$ to $\mathbf{r}$
then we have (reflecting conservation of commodities):%
\[
\int_{T}\boldsymbol{a=}\int_{T}\mathbf{r}%
\]

We wrap the \emph{Aggregation} and \emph{Anonymity} conditions into one, and
directly postulate that the return to any individual depends only on his own
offer and the integral of everyone's offers, and that this return function is
the same for everyone. Thus we have a function $r$ from $S\times S_{+}$ to $C$
such that $\mathbf{r(}t\mathbf{)}=r(a,b)$, where $a=\mathbf{a(}t\mathbf{)}$
and $b=\int_{T}\boldsymbol{a}.$ The following lemma is essentially from
\cite{Dubey-MasColell-Shubik:1980}.

\begin{proposition}
$r(a,b)$ is linear in $a$ (for fixed $b$) and $r(a,\lambda b)=r(a,b)$ for any
$a,b$ and positive scalar $\lambda.$
\end{proposition}

\begin{proof}
We will first show that if $a,c$ $\in$ $S$ and $0<\lambda<1$, then
\[
r(\lambda a+(1-\lambda)c,b)=\lambda r(a,b)+(1-\lambda)r(c,b)
\]
There clearly exists an integrable map $\mathbf{d}$ from $T=[0,1]$ to space of
offers $S$ such that (i) positive mass of traders choose $a$ in $\mathbf{d}$;
(ii) positive mass of traders choose $c$ in $\mathbf{d}$ ; and (iii) the
integral of $\mathbf{d}$ on $T$ is $b.$ So $\int_{T}r(\mathbf{d}^{\alpha
}\mathbf{,}b)d\mu(\alpha)=\int_{T}r(\mathbf{d,}b)=\overline{b}$ since
commodities are conserved. Shift $\varepsilon\lambda$ mass from $a$ to
$\lambda a+(1-\lambda)c$ and $(1-\lambda)\varepsilon$ mass from $c$ to
$\lambda a+(1-\lambda)c$ , letting the rest be according to $\mathbf{d}.$ This
yields a new function (from $T$ to $S$ ) which we call $\mathbf{e}.$ Clearly
the integral of $\mathbf{e}$ on $T$ is also $b.$ Therefore, once again by
conservation of commodities, we must have $\int_{T}r(\mathbf{e,}%
b)=\overline{b},$ hence $\int_{T}r(\mathbf{d,}b)=$ $\int_{T}r(\mathbf{e,}b).$
But this can only be true if the displayed equality holds, proving that (every
coordinate of) $r$ is affine in $a$ for fixed $b$.

Now $r(0,b)\geq0$ by assumption. Suppose $r(0,b)\gvertneqq0$. Partition $T$
into two non-null sets $T_{1}$ and $T_{2}.$ Consider the case where all the
individuals in $T_{1}$ offer $0,$ and all in $T_{2}$ offer $b/\mu(T_{2}).$
Then, since everone in $T_{1\text{ }}$gets the return $r(0,b)\gvertneqq0,$ by
conservation of commodities everyone in $T_{2}$ gets $\overline{b}-\mu(T_{1})$
$r(0,b)\lvertneqq b/\mu(T_{2}),$ contradicting \emph{non-dissipation. }So
$r(0,b)=0,$ showing $r$ is linear.

Finally $\lambda r(a,b)=r(\lambda a,\lambda b)=\lambda r(a,\lambda b),$where
the first equality comes from \emph{Invariance} and the second from linearity.
\end{proof}

\begin{remark}
As mentioned in the introduction, when there is a continuum of traders, the
star mechanism leads to equivalence (or, near-equivalence) of Nash and Walras
equilibria under suitable postulates regarding the commodity or fiat money.
(See \cite{Dubey-Shapley: 1994} for a detailed discussion.)\ 
\end{remark}

\end{document}